\documentclass[pra,aps,showpacs,twocolumn,twoside,superscriptaddress]{revtex4}

\usepackage{amsmath,amsfonts,amssymb,color,epsfig,graphics,graphicx,latexsym,revsymb,theorem,verbatim}

\newtheorem{definition}{Definition}
\newtheorem{proposition}[definition]{Proposition}
\newtheorem{lemma}[definition]{Lemma}

\newtheorem{theorem}[definition]{Theorem}
\newtheorem{corollary}[definition]{Corollary}
\newtheorem{conjecture}[definition]{Conjecture}

\newtheorem{remark}[definition]{Remark}
\newtheorem{example}[definition]{Example}

\def\squareforqed{\hbox{\rlap{$\sqcap$}$\sqcup$}}
\def\qed{\ifmmode\squareforqed\else{\unskip\nobreak\hfil
\penalty50\hskip1em\null\nobreak\hfil\squareforqed
\parfillskip=0pt\finalhyphendemerits=0\endgraf}\fi}
\def\endenv{\ifmmode\;\else{\unskip\nobreak\hfil
\penalty50\hskip1em\null\nobreak\hfil\;
\parfillskip=0pt\finalhyphendemerits=0\endgraf}\fi}
\newenvironment{proof}{\noindent \textbf{{Proof.~} }}{\qed}

\def\bcj{\begin{conjecture}}
\def\ecj{\end{conjecture}}
\def\bcr{\begin{corollary}}
\def\ecr{\end{corollary}}
\def\bd{\begin{definition}}
\def\ed{\end{definition}}
\def\bea{\begin{eqnarray}}
\def\eea{\end{eqnarray}}
\def\bem{\begin{enumerate}}
\def\eem{\end{enumerate}}
\def\bex{\begin{example}}
\def\eex{\end{example}}
\def\bim{\begin{itemize}}
\def\eim{\end{itemize}}
\def\bl{\begin{lemma}}
\def\el{\end{lemma}}
\def\bpf{\begin{proof}}
\def\epf{\end{proof}}
\def\bpp{\begin{proposition}}
\def\epp{\end{proposition}}
\def\br{\begin{remark}}
\def\er{\end{remark}}
\def\bt{\begin{theorem}}
\def\et{\end{theorem}}

\def\r{\rho}
\def\s{\sigma}

\def\ph{\varphi}

\def\ps{\psi}

\def\Ps{\Psi}

\newcommand{\nc}{\newcommand}

\nc{\cA}{{\cal A}} \nc{\cB}{{\cal B}} \nc{\cC}{{\cal C}}
\nc{\cD}{{\cal D}} \nc{\cE}{{\cal E}} \nc{\cF}{{\cal F}}
\nc{\cG}{{\cal G}} \nc{\cH}{{\cal H}} \nc{\cI}{{\cal I}}
\nc{\cJ}{{\cal J}} \nc{\cK}{{\cal K}} \nc{\cL}{{\cal L}}
\nc{\cM}{{\cal M}} \nc{\cN}{{\cal N}} \nc{\cO}{{\cal O}}
\nc{\cP}{{\cal P}} \nc{\cR}{{\cal R}} \nc{\cS}{{\cal S}}
\nc{\cT}{{\cal T}} \nc{\cV}{{\cal V}} \nc{\cX}{{\cal X}}
\nc{\cZ}{{\cal Z}}

\def\Tr{\mathop{\rm Tr}\nolimits}

\def\max{\mathop{\rm max}}

\def\rank{\mathop{\rm rank}}
\def\rk{\mathop{\rm rk}}

\def\tr{\mathop{\rm Tr}}

\def\bigox{\bigotimes}

\def\ox{\otimes}
\def\ra{\rightarrow}

\newcommand{\bra}[1]{\langle#1|}
\newcommand{\ket}[1]{|#1\rangle}
\newcommand{\proj}[1]{| #1\rangle\!\langle #1 |}
\newcommand{\ketbra}[2]{|#1\rangle\!\langle#2|}

\newcommand{\cmp}{Comm. Math. Phys.}

\newcommand{\pla}{Phys. Lett. A}

\begin{document}
\title{Non-distillable entanglement guarantees distillable entanglement}

\author{Lin Chen}
\email{cqtcl@nus.edu.sg(Corresponding~Author)}
\address{Department of Pure Mathematics and Institute for Quantum Computing, University
of Waterloo, Waterloo, Ontario, N2L 3G1, Canada}
\address{Centre for Quantum Technologies, National
University of Singapore, 3 Science Drive 2, Singapore 117542} 

\author{Masahito Hayashi}
\email{hayashi@math.is.tohoku.ac.jp}
\address{Graduate School of Information Sciences, Tohoku
University, Aoba-ku, Sendai, 980-8579, Japan}
\address{Centre for Quantum Technologies, National University of Singapore, 3 Science Drive 2, 117542, Singapore}

\begin{abstract}
The monogamy of entanglement is one of the basic quantum mechanical
features, which says that when two partners Alice and Bob are more
entangled then either of them has to be less entangled with the
third party. Here we qualitatively present the converse monogamy of
entanglement: given a tripartite pure system and when Alice and Bob
are entangled and non-distillable, then either of them is
distillable with the third party. Our result leads to the
classification of tripartite pure states based on bipartite reduced
density operators, which is a novel and effective way to this
long-standing problem compared to the means by stochastic local
operations and classical communications. Furthermore we
systematically indicate the structure of the classified states and
generate them. We also extend our results to multipartite states.

\end{abstract}

\date{ \today }

\pacs{03.65.Ud, 03.67.Mn, 03.67.-a}



\maketitle
\section{Introduction}

The monogamy of entanglement is a purely quantum phenomenon in
physics \cite{ckw00} and has been used in various applications, such
as bell inequalities \cite{sg01} and quantum security \cite{KW}. In
general, it indicates that the more entangled the composite system
of two partners Alice $(A)$ and Bob $(B)$ is, the less entanglement
between $A$ $(B)$ and the environment $E$ there is. The security of
many quantum secret protocols can be guaranteed quantitatively
\cite{KW,TH}. However the converse statement generally doesn't hold,
namely when $A$ and $B$ are less entangled, we cannot decide whether
$A$ $(B)$ and $E$ are more entangled. In fact even when the formers
are classically correlated namely separable \cite{werner89}, the
latters may be also separable. For example, this is realizable by
the tripartite Greenberger-Horne-Zeilinger (GHZ) state.

Nevertheless, it is still important to {\em qualitatively}
characterize the above converse statement in the light of the
hierarchy of entanglement of bipartite systems. Such a
characterization defines a converse monogamy of entanglement, and
there is no classical counterpart. Besides, it is also expected to
be helpful for treating a quantum multi-party protocol when the
third party helps the remaining two parties, for it guarantees the
property of one reduced density operator from another. To justify
the hierarchy of entanglement, we recall six well-known conditions,
i.e., the separability, positive-partial-transpose (PPT)
\cite{hhh96,HHH}, non-distillability of entanglement under local
operations and classical communications (LOCC) \cite{HHH}, reduction
property (states satisfying reduction criterion) \cite{HH},
majorization property \cite{Hiro} and negativity of conditional
entropy \cite{HH}. These conditions form a hierarchy since a
bipartite state satisfying the former condition will satisfy the
latter too. Therefore, the strength of entanglement in the states
satisfying the conditions in turn becomes gradually {\em weak}.

For example, the hierarchy is closely related to the distilability
of entanglement \cite{HHH}. While PPT entangled states cannot be
distilled to Bell states for implementing quantum information tasks,
Horodecki's protocol can distill a state that violates reduction
criterion \cite{HH}. That is, the former entangled state is useless
as a resource while the latter entangled state is useful. So the
usefulness of entangled states can be characterized by this
hierarchy. Recently, a hierarchy of entanglement has been developed
based on these criteria \cite{hc11}.

\begin{figure}
  \includegraphics[width=3cm]{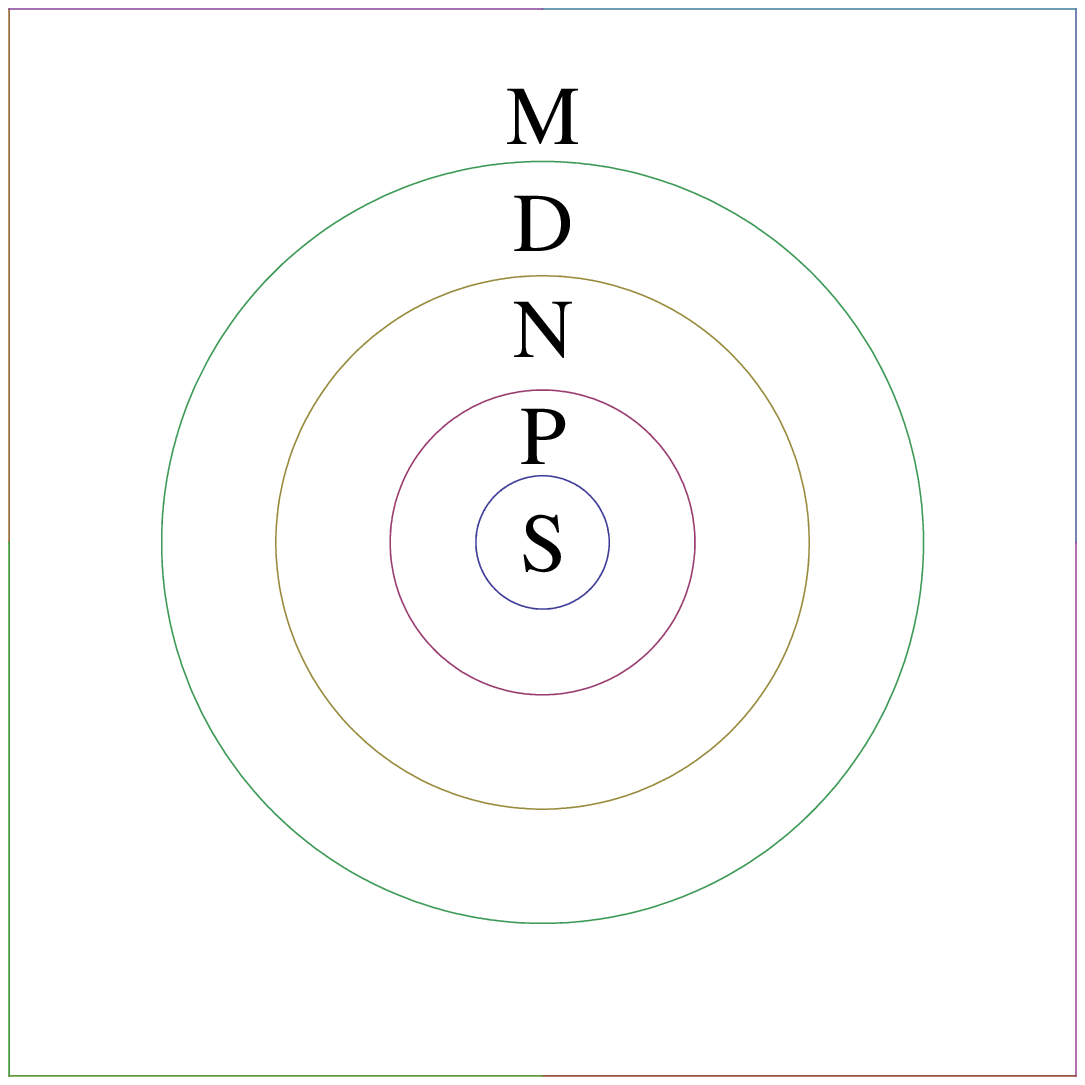}\\
  \caption{\label{fig:SPNDM} Hierarchy of bipartite states in terms of five sets S, P, N, D, M.
  Intuitively, the sets S and P form all
  PPT states, the sets S, P, N, D and M form all states satisfying and violating
  the reduction criterion, respectively. So the five sets
  constitute the set of bipartite states and there is no intersection between any two sets. The strength of
  entanglement of the five sets becomes weak in turn, $S \le P \le N \le D  \le M$.}
\end{figure}

In this paper, for simplicity we consider four most important
conditions, namely the separability, PPT, non-distillability and the
reduction criteria. Then we establish a further hierarchy of
entanglement consisting of five sets: separable states (S),
non-separable PPT states (P), non-PPT non-distillable states (N),
distillable reduction states (D), and non-reduction states (M), see
Figure \ref{fig:SPNDM}. In particular the states belonging to $M$
are always distillable \cite{HH}.

We show that when the entangled state between $A$ and $B$, i.e.,
$\r_{AB}$ belongs to the set $D$, then the state between $A$ $(B)$
and $E$, i.e., $\r_{AE} (\r_{BE})$ belongs to the set $D$ or $M$.
Likewise when $\r_{AB}$ belongs to $P$ or $N$, then $\r_{AE}
(\r_{BE})$ must belong to $M$. Hence we can qualitatively
characterize the converse monogamy of entanglement as follows: when
the state $\r_{AB}$ is weakly entangled, then $\r_{AE}$ is generally
strongly entangled in terms of the five sets $S,P,N,D,M$. These
assertions follow from a corollary of Theorem
\ref{thm:sixconditions} to be proved later.
\begin{theorem}
\label{thm1} Suppose a tripartite pure state has a non-distillable
bipartite reduced state. Then another bipartite reduced state is
separable if and only if it satisfies the reduction criterion.
\end{theorem}
From this theorem, we will solve two conjectures on the existence of
specified tripartite state proposed by Thapliyal in 1999
\cite{thapliyal99}. On the other hand, the theorem also helps
develop the classification of tripartite pure states based on the
three reduced states, each of which could be in one of the five
cases $S,P,N,D,M$. So there are at most $5^3=125$ different kinds of
tripartite states. Evidently, some of them do not exist due to
Theorem \ref{thm1}. It manifests that the quantum behavior of a
global system is strongly restricted by local systems. By
generalizing to many-body systems, we can realize the quantum nature
on macroscopic size in terms of the microscopic physical systems.
This is helpful to the development of matter and material physics
\cite{hw11}. Hence, in theory it becomes important to totally
identify different tripartite states.

To explore the problem, we describe the properties for reduced
states $\rho_{AB}$, $\rho_{BC}$, and $\rho_{CA}$ of the state
$|\Psi\rangle$ by $X_{AB},X_{BC},X_{CA}$ that take values in
$S,P,N,D,M$. The subset of such states $|\Psi\rangle$ is denoted by
$\cS_{X_{AB}X_{BC}X_{CA}}$, and the subset is non-empty when there
exists a tripartite state in it. For example, the GHZ state belongs
to the subset $\cS_{SSS}$. Furthermore as is later shown in Table
\ref{tab:tripartite}, $|\Psi\rangle$ belongs to the subset
$\cS_{SSM}$ when the reduced state $\rho_{CA}$ is an entangled
maximally correlated state \cite{rains99}. By Theorem \ref{thm1},
one can readily see that any non-empty subset is limited in nine
essential subsets,
$$\cS_{SSS},\cS_{SSM},\cS_{SMM},\cS_{PMM},\cS_{NMM},$$
and
$$\cS_{DDD},\cS_{DDM},\cS_{DMM},\cS_{MMM}$$.
Hence up to permutation, the number of non-empty subsets for
tripartite pure states is at most $21=1\times 3+3\times 6$. In
particular, it is a long-standing open problem that whether
$\cS_{NMM}$ exists \cite{dss00}. Except tho subsets generated from
$\cS_{NMM}$, we will demonstrate that the rest 18 subsets are indeed
non-empty by explicit examples. These subsets are not preserved
under the conventional classification by the invertible stochastic
LOCC (SLOCC) \cite{dvc00,cds08,ccd10}. We will explain these results
in Sec. III.

Furthermore, we show that the subsets form a commutative monoid and
it is a basic algebraic concept. We systematically characterize the
relation of the subsets by generating them under the rule of monoid
in the late part of Sec. III.

We also extend our results to multipartite scenario. In particular,
we introduce the multipartite separable, PPT and non-distillable
states. They become pairwise equivalent when they are compatible to
a pure state, see Theorem \ref{thm:ppt,n}, Sec. IV. Finally we
conclude in Sec. V.


\section{Unification of entanglement criterion}

In quantum information, the following six criteria are extensively
useful for studying bipartite states $\r_{AB}$ in the space $\cH_A
\ox \cH_B$.
\begin{description}
\item[(1)]
Separability: $\rho_{AB}$ is the convex sum of product states
\cite{werner89}.

\item[(2)]
PPT condition: the partial transpose of $\rho_{AB}$ is semidefinite
positive \cite{hhh96}.

\item[(3)]
Non-distillability: no pure entanglement can be asymptotically
extracted from $\rho_{AB}$ under LOCC, no matter how many copies are
available \cite{HHH}.

\item[(4)]
Reduction criterion: $\rho_{A}\otimes I_B \ge \rho_{AB}$ and $I_A
\otimes \rho_{B} \ge \rho_{AB}$ \cite{HH}.

\item[(5)]
Majorization criterion: $\rho_{A} \succ \rho_{AB}$ and $\rho_{B}
\succ \rho_{AB}$ \cite{Hiro}.

\item[(6)]
Conditional entropy criterion: $H_{\rho}(B|A)=  H(\rho_{AB})-
H(\rho_{A}) \ge 0$ and $H_{\rho}(A|B)=  H(\rho_{AB})- H(\rho_{B})
\ge 0$, where $H$ is the von Neumann entropy.
\end{description}
It is well-known that the relation $ {\bf (1)} \Rightarrow {\bf (2)}
\Rightarrow {\bf (3)} \Rightarrow {\bf (4)} \Rightarrow {\bf (5)}
\Rightarrow {\bf (6)}$ holds for any state $\r_{AB}$
\cite{HHH,HH,Hiro}. In particular apart from ${\bf (2)} \Rightarrow
{\bf (3)}$ whose converse is a famous open problem \cite{dss00}, all
other relations are strict. We will show that these conditions
become equal when we further require $\rho_{BC}$ is non-distillable.
First, under this restriction the conditions {\bf (5)} and {\bf (6)}
are respectively simplified into {\bf(5')} $\rho_{A} \sim
\rho_{AB}$, where $\sim$ denotes that $\rho_{A}$ and $\rho_{AB}$
have identical eigenvalues, and {\bf(6')} $H(\rho_{A}) =
H(\rho_{AB})$. Second, when $\rho_{BC}$ is non-distillable, since
$\rho_{AB} \succ \rho_{A}$ holds, the above two conditions {\bf
(5')} and {\bf (6')} are equivalent. Now we have

\begin{theorem}
\label{thm:sixconditions}
 For a tripartite state $\ket{\Ps}_{ABC}$ with
 a non-distillable reduced state $\rho_{BC}$ namely condition {\bf (3)}, then conditions {\bf (1)}-{\bf (6)},
 {\bf (5')}, and {\bf (6')} are equivalent for $\rho_{AB}$.
\end{theorem}
The proof is given in the appendix of this paper. We can readily get
Theorem \ref{thm1} from Theorem \ref{thm:sixconditions}, and provide
its operational meaning as the main result of the work.

 \bt
 \label{thm:conversemonogamy}
 $\bf{(Converse~monogamy~of~entanglement).}$ Consider a tripartite
 state $\ket{\Ps}_{ABE}$ with entangled reduced states $\r_{AB}, \r_{AE}$ and
 $\r_{BE}$. When $\r_{AB}$ is a weakly entangled state $P$ or $N$
 ($D$), the states $\r_{AE}$ and $\r_{BE}$ are strongly
 entangled states $M$ ($D$ or $M$).
 \et
To our knowledge, the converse monogamy of entanglement is another
basic feature of quantum mechanics and there is no classical
counterpart since classical correlation can only be "quantified". In
contrast, quantum entanglement has qualitatively different levels of
strength and they have essentially different usefulness from each
other. For example the states in the subset $N$ cannot be distilled
while those in $M$ are known to be distillable \cite{HH}. So only
the latter can directly serve as an available resource for quantum
information processing and it implies the following paradox. {\bf A
useless entangled state between $A$ and $B$ strengthens the
usefulness of entanglement resource between $A$ ($B$) and the
environment}. Therefore, the converse monogamy of entanglement
indicates a dual property to the monogamy of entanglement: Not only
the amount of entanglement, the usefulness of entanglement in a
composite system is also restricted by each other.

Apart from bringing about the converse monogamy of entanglement,
Theorem \ref{thm:sixconditions} also promotes the study over a few
important problems. For instance, the equivalence of {\bf (1)} and
{\bf (2)} is a necessary and sufficient condition of deciding
separable states, beyond that for states of rank not exceeding 4
\cite{hhh96,hlv00,cd11}. Besides, the equivalence of {\bf (2)} and
{\bf (3)} indicates another kind of non-PPT entanglement activation
by PPT entanglement \cite{evw01}. For later convenience, we
explicitly work out the expressions of states satisfying the
assumptions in Theorem \ref{thm:sixconditions}.
 \bl
 \label{le:SNSSSS}¡¡
 The tripartite pure state with two non-distillable reduced states $\r_{AB}$ and $\r_{AC}$,
 if and only if it has the form $\sum_i \sqrt{p_i}
 \ket{b_i,i,i}$ up to local unitary operators. In other words, the
 reduced state $\r_{BC}$ is maximally correlated.
 \el
For the proof see Lemma 11 in \cite{cd11}. We apply our results to
handle two open problems in FIG. 4 of \cite{thapliyal99}, i.e.,
whether there exist tripartite states with two PPT bound entangled
reduced states, and tripartite states with two separable and one
bound entangled reduced states. Here we give negative answers to
these open problems in terms of Theorem~\ref{thm:sixconditions} and
Lemma \ref{le:SNSSSS}. As the first problem is trivial, we account
for the second conjecture. Because the required states have the form
$\sum^d_{i=1} \sqrt{p_i}\ket{ii}\ket{c_i}$, where $\r_{AC}$ and
$\r_{BC}$ are separable. So the reduced state $\r_{AB}$ is a
maximally correlated state, which is either separable or
distillable. It readily denies the second problem.

It is noticeable that the converse of Theorem
\ref{thm:sixconditions} doesn't hold. That is, for a tripartite
state $\ket{\Ps}_{ABC}$ for which conditions {\bf (1)}-{\bf (6)},
{\bf (5')}, and {\bf (6')} are equivalent for $\rho_{AB}$, the
reduced state $\rho_{BC}$ is not necessarily non-distillable. An
example is the state $\ket{000}+(\ket{0}+\ket{1})\ket{11}$.

Finally we extend Theorem \ref{thm:sixconditions} for the tripartite
state containing a qubit reduced state. For this purpose we
introduce the following known result \cite[Remark 1]{cag99}.

 \bl
 \label{le:2xN,PPT=Reductioncriterion}
A $2\times N$ state is PPT if and only if it satisfies the reduction
criterion.
 \el
Based on it we have
 \bt
 \label{thm:sixconditions,qubit}
Suppose $\ket{\Ps}_{ABC}$ contains a qubit reduced state. If
$\rho_{BC}$ satisfies condition {\bf (4)}, then conditions {\bf
(1)}-{\bf (6)}, {\bf (5')}, and {\bf (6')} are equivalent for
$\rho_{AB}$.
 \et
 \bpf
When $\r_{BC}$ contains a qubit reduced state, the claim follows
from Lemma \ref{le:2xN,PPT=Reductioncriterion} and Theorem
\ref{thm:sixconditions}. So it suffices to consider the case
$\rank\r_A=2$. Since $\rho_{BC}$ satisfies condition {\bf (4)}, we
obtain $\rank\r_A \ge \rank\r_B, \rank\r_C$. In other word $\r_{BC}$
is a two-qubit state and the claim again follows from Lemma
\ref{le:2xN,PPT=Reductioncriterion} and Theorem
\ref{thm:sixconditions}. This completes the proof.
 \epf
For tripartite states with higher dimensions, Theorem
\ref{thm:sixconditions,qubit} does not apply anymore and we will see
available examples in next section. As the concluding remark of this
section, we propose the following conjecture.
 \bcj
For a tripartite $3\times3\times3$ state $\ket{\Ps}_{ABC}$, suppose
$\r_{BC}$ satisfies {\bf (4)} and $\r_{AB}$ satisfies {\bf (5)}.
Then $\r_{AB}$ also satisfies {\bf (4)}.
 \ecj

\section{Classification with reduced states}

Theorem \ref{thm:conversemonogamy} says that the quantum correlation
between two parties of a tripartite system is dependent on the third
party. From Theorem \ref{thm:sixconditions} and the discussion to
conjectures in \cite{thapliyal99}, we can see that the tripartite
pure state with some specified bipartite reduced states may not
exist. This statement leads to a classification of tripartite states
in terms of the three reduced states \cite{sa08}. As a result, we
obtain the different subsets of tripartite states in Table
\ref{tab:tripartite} in terms of tensor rank and local ranks of each
one-party reduced state. In the language of quantum information, the
tensor rank of a multipartite state, also known as the Schmidt
measure of entanglement~\cite{eb01}, is equal to the least number of
product states to expand this state. For instance, any multiqubit
GHZ state has tensor rank two. So tensor rank is bigger or equal to
any local rank of a multipartite pure state. As tensor rank is
invariant under invertible SLOCC \cite{cds08}, it has been widely
applied to classify SLOCC-equivalent multipartite states recently
\cite{ccd10}.

Here we will see that, tensor rank is also essential to the
classification in Table \ref{tab:tripartite}. We will justify the
statement for each subset in Table \ref{tab:tripartite}, then we
show their nonemptyness by proposing specific examples.

First the statement for the subsets $\cS_{SSS},\cS_{SSM}$ and
$\cS_{SMM}$ follows from Lemma \ref{le:SNSSSS}, and Lemma 2 in
\cite{dfx07}, respectively. The nonemptyness readily follows from
the states $\ket{\ps}_{ABC}$ in Table \ref{tab:tripartite}. To
verify the statement for $\cS_{PMM}$ and $\cS_{NMM}$, we propose the
following result

\begin{lemma}
  \label{le:rankAB=rankA,B}
  A $M \times N$ state with rank $N$ is non-distillable if and only
  if it is separable and is the convex sum of just $N$ product states, i.e.,
  $\sum^N_{i=1} p_i \proj{a_i,b_i}$.
\end{lemma}
 \bpf
It suffices to show the necessity. Let $\r$ be a $M \times N$ state
with rank $N$ and suppose it is non-distillable. From Lemma 11 in
\cite{cd11} we obtain $\r$ is PPT. The claim then follows from
\cite{hlv00}. This completes the proof.
 \epf

As a result, the PPT entangled state $\r_{AB}$ satisfies
$d_A,d_B<\rank \r_{AB}$. So the purification of PPT entangled states
$\r_{AB}$ is subject to the statement for subset $\cS_{PMM}$ in
Table \ref{tab:tripartite}. It also shows the nonemptyness of
$\cS_{PMM}$. A similar argument can be used to justify the statement
for $\cS_{NMM}$ in Table \ref{tab:tripartite}, if it really exists.

Next we study the subset $\cS_{DDD}$. Since the reduced state
$\r_{AB}$ satisfies the reduction criterion, we have $d_C \ge d_A,
d_B $. By applying the same argument to other reduced states we
obtain $d_A=d_B=d_C$. Since $\r_{AB}$ is distillable, the rest
statement $r>d_A$ in Table \ref{tab:tripartite} follows from the
following observation.

\begin{lemma}
\label{ha-l-2} Assume that $\rk(\Ps) = \max \{ d_A, d_B \}$. Then
the conditions {\bf (1)}-{\bf (4)} are equivalent for $\rho_{AB}$.
\end{lemma}
\begin{proof}
It suffices to show that the state $\rho_{AB}$ is separable when it
satisfies the reduction criterion. Let $|\Psi
\rangle=\sum_{i=1}^{\rk(\Ps)} \sqrt{p_i} |a_i,b_i,c_i\rangle$, and
$V_A$ an invertible matrix such that $V_A|i\rangle=|a_i\rangle$.
Now, we focus on the pure state $|\Psi' \rangle= K V_A^{-1} \otimes
I_B\otimes I_C |\Psi \rangle$, where $K$ is the normalized constant.
Then the reduced state $\rho_{AB}'$ satisfies $\rho_A' \otimes I_B
\ge \rho_{AB}'$, and hence $\rho_A' \succ \rho_{AB}'$ \cite{Hiro}.
Since $\rho_{BC}'$ is separable, we have $\rho_A' \sim \rho_{AB}'$.
So the state $\rho_{AB}'$, and equally $ \rho_{AB} $  is separable
in terms of Theorem \ref{thm:sixconditions}.
\end{proof}

 \bex
It's noticeable that under the same assumption in Lemma
\ref{ha-l-2}, the equivalence between conditions {\bf (1)}-{\bf (5)}
does not hold. As a counterexample, we consider the symmetric state
$|\Psi\rangle = \frac{1}{\sqrt{r+7}}(\sum_{i=2}^r|iii\rangle +
(|1\rangle +|2 \rangle )(|1\rangle +|2 \rangle )(|1\rangle +|2
\rangle ))$. It is symmetric and thus satisfies condition {\bf (5)}.
On the other hand one can directly show that any reduced state of
$|\Psi\rangle$ violate the reduction criterion, so it does not
satisfy conditions {\bf (1)}-{\bf (4)}. Hence $|\Psi\rangle$ belongs
to the subset ${\cS}_{MMM}$. It indicates that tensor rank alone is
not enough to characterize the hierarchy of bipartite entanglement.
 \eex

 \bex
The symmetric state $\ket{\ps_r}=
\frac{1}{\sqrt{2r}}(|312\rangle+|123\rangle+|231\rangle+|213\rangle+|132\rangle+|321\rangle)
+\frac{1}{\sqrt{r}} \sum_{j=4}^r|jjj\rangle$ indicates the
nonemptyness of $ \cS_{DDD}$ for $d_A \ge 3$. To see it, it suffices
to show one of the reduced states, say $\r_{AB}$ satisfies the
reduction criterion and is distillable simultaneously. The former
can be directly justified. By performing the local projector
$P=\proj{1}+\proj{2}$ on system A, we obtain the resulting state
$P\ox I \r_{AB} P\ox I = (\ket{12}+\ket{21})(\bra{12}+\bra{21})$
which is a Bell state. So $\r_{AB}$ is distillable.
 \eex
However there is no state with $d_A=2$ in $ \cS_{DDD}$. The reason
is that a two-qubit state satisfying the reduction criterion is also
PPT by Lemma \ref{le:2xN,PPT=Reductioncriterion}. Hence it must be
separable in terms of Peres' condition \cite{peres96}. It
contradicts with the statement that $\r_{AB}$ is distillable.

By using a similar argument to $\cS_{DDD}$, one can verify the
statement for $\cS_{DDM}$ in Table \ref{tab:tripartite}. A concrete
example will be built by the rule of monoid later. Hence $\cS_{DDM}$
is nonempty.

Third, we characterize the subset $\cS_{DMM}$ by the tensor rank of
states in this subset. It follows from the definition of reduction
criterion that $d_C \ge d_A, d_B $. This observation and Lemma
\ref{ha-l-2} justify the statement for $\cS_{DMM}$ in Table
\ref{tab:tripartite}. In order to show the tightness of these
inequalities, we consider the non-emptyness of the subset
$\cS_{DMM}$ with the boundary types $r = d_C > d_A= d_B$ and $r >
d_C = d_A = d_B$.
 \bex
To justify the former type, it suffices to consider the state
$\ket{\ps_a}=\frac{1}{\sqrt{6+3a^2}}(\ket{123}+\ket{231}+\ket{312}
+\ket{21}(\ket{3}+a\ket{6})+\ket{13}(\ket{2}+a\ket{5})+\ket{32}(\ket{1}+a\ket{4})),a\in
R$. It obviously fulfils $r = d_C > d_A= d_B$, so the reduced states
$\r_{AC}$ and $\r_{BC}$ violate the reduction criterion. Next we
focus on the reduced state $\r_{AB}$. By performing the local
projector $P=\proj{1}+\proj{2}$ on system A, we obtain the resulting
state $P\ox I \r_{AB} P\ox I =
(\ket{12}+\ket{21})(\bra{12}+\bra{21})+a^2\proj{21}$. This is an
entangled two-qubit state and is hence distillable \cite{hhh97}. On
the other hand to see that $\r_{AB}$ fulfils the reduction criterion
for any $a$, one need notice the facts $\r_A=\r_B=\frac13I$ and the
eigenvalues of $\r_{AB}$ are not bigger than $\frac13$.
 \eex
In addition, an example of the latter type is constructed by the
rule of monoid later. Thus, we can confirm the tightness of the
above inequalities of two boundary types for $\cS_{DMM}$.

To conclude, we have verified the statement and existence of all
essential subsets in Table \ref{tab:tripartite} except the
$\cS_{NMM}$.



\begin{widetext}

\begin{table}
 \caption{\label{tab:tripartite}
 Classification of tripartite states
 $\ket{\ps}_{ABC}$ in terms of the bipartite reduced states.
 The table contains neither the classes generated from the
 permutation of parties, and nor the subset $\cS_{MMM}$ since for which there is no
 fixed relation between the tensor rank $\rk(\ps)$ and
 {\em local ranks} $d_A(\Psi)$, $d_B(\Psi)$,
 and $d_C(\Psi)$. They are simplified to $r$, $d_A$, $d_B$, and $d_C$ when there is no confusion. All expressions
 are up to local unitaries (LU) and all sums run from 1 to $r$.
 In the subset $\cS_{SMM}$, the linearly independent states $\ket{c_i}$ are the support of space $\cH_C$.
 Apart from the subset $\cS_{NMM}$, all other eight subsets are nonempty in terms of specific examples in Sec. III. }
{
  \begin{tabular}{ | c| c| c|c|}
    \hline
    $\cS_{X_{AB}X_{BC}X_{CA}}$
    & $\cS_{SSS}$
    & $\cS_{SSM}$
    & $\cS_{SMM}$
    \\
    \hline
    \begin{tabular}{c}
    expression of
    \\
    $\ket{\ps}_{ABC}$
    \end{tabular}
    & $\displaystyle \small \sum_i \sqrt{p_i} \ket{iii}$
    & $\displaystyle \small \sum_i \sqrt{p_i} \ket{i,b_i,i}$
    & $\displaystyle \small \sum_i \sqrt{p_i} \ket{a_i,b_i,c_i}$
    \\
    \hline
    \begin{tabular}{c}
    tensor rank
    \\and
    \\local ranks
    \end{tabular}
    &
    \begin{tabular}{c}
    $r=d_A=$\\
    $d_B=d_C$
    \end{tabular}
    &
    \begin{tabular}{c}
    $r=d_A=$\\
    $d_C \ge d_B$
    \end{tabular}
    &
    \begin{tabular}{c}
    $r=$\\
    $d_C \ge d_A, d_B$
    \end{tabular}
    \\
    \hline
  \end{tabular}
  }

  \begin{tabular}{ | c| c| c| c|}
    \hline
    $\cS_{X_{AB}X_{BC}X_{CA}}$
    & $\cS_{PMM}$
    & $\cS_{NMM}$
    & $\cS_{DDD}$
    \\
    \hline
    \begin{tabular}{c}
    expression of
    \\
    $\ket{\ps}_{ABC}$
    \end{tabular}
    & $\displaystyle \small \sum_i \sqrt{p_i} \ket{a_i,b_i,c_i}$
    & $\displaystyle \small \sum_i \sqrt{p_i} \ket{a_i,b_i,c_i}$
    & $\displaystyle \small \sum_i \sqrt{p_i} \ket{a_i,b_i,c_i}$
    \\
    \hline
    \begin{tabular}{c}
    tensor rank
    \\and
    \\local ranks
    \end{tabular}
    &
    \begin{tabular}{c}
    $r \ge$\\
    $d_C > d_A, d_B$
    \end{tabular}
    &
    \begin{tabular}{c}
    $r \ge$\\
    $d_C > d_A, d_B$
    \end{tabular}
    &
    \begin{tabular}{c}
    $r >$\\
    $d_C = d_B = d_A$
    \end{tabular}
    \\
    \hline
  \end{tabular}

  \begin{tabular}{ | c| c| c|}
    \hline
    $\cS_{X_{AB}X_{BC}X_{CA}}$
    & $\cS_{DDM}$
    & $\cS_{DMM}$
    \\
    \hline
    \begin{tabular}{c}
    expression of
    \\
    $\ket{\ps}_{ABC}$
    \end{tabular}
    & $\displaystyle \small \sum_i \sqrt{p_i} \ket{a_i,b_i,c_i}$
    & $\displaystyle \small \sum_i \sqrt{p_i} \ket{a_i,b_i,c_i}$
    \\
    \hline
    \begin{tabular}{c}
    tensor rank
    \\and
    \\local ranks
    \end{tabular}
    &
    \begin{tabular}{c}
    $r >$\\
    $d_C = d_A \ge d_B$
    \end{tabular}
    &
    \begin{tabular}{c}
    $r \ge d_C \ge d_A, d_B$\\
    $r > d_A, d_B$
    \end{tabular}
    \\
    \hline
  \end{tabular}

\end{table}
\end{widetext}

\medskip
\noindent \textit{Comparison to SLOCC classification}. We know that
there are much efforts towards the classification of multipartite
state by invertible SLOCC \cite{dvc00,cds08,ccd10}. Hence, it is
necessary to clarify the relation between this method and the
classification by reduced states in Table \ref{tab:tripartite}. When
we adopt the former way we have no clear characterization to the
hierarchy of bipartite entanglement between the involved parties,
i.e., the structure of reduced states becomes messy under SLOCC. Our
classification resolves this drawback. Another potential advantage
of our idea is that we can apply the known fruitful results of
bipartite entanglement, such as the hierarchy of entanglement to
further study the classification problem.

Here we explicitly exemplify that the invertible SLOCC only
partially preserves the classification in Table
\ref{tab:tripartite}. We focus on the orbit $\cO_{r=d_A=d_B=d_C}:=
\{ |\Psi\rangle | \rk(\Psi)= d_A(\Psi)=d_B(\Psi)=d_C(\Psi) \}$,
which has non-empty intersection with the subsets $\cS_{SSS}$,
$\cS_{SSM}$, and $\cS_{SMM}$. Further, since the subset $\cS_{MMM}$
contains the state $|\Psi_a\rangle$, it also has non-empty
intersection with $\cO_{r=d_A=d_B=d_C}$. Since all state in
$\cO_{r=d_A=d_B=d_C}$ can be converted to GHZ state by invertible
SLOCC, it does not preserve the classification by reduced states.
However, the subsets $\cS_{DDM}$ and $\cS_{DDD}$ are not mixed with
$\cS_{SSS}$, $\cS_{SSM}$, and $\cS_{SMM}$ by invertible SLOCC.

\medskip
\noindent \textit{Monoid structure}. To get a further understanding
of Table \ref{tab:tripartite} from the algebraic viewpoint, 
we define the direct sum for subsets by
$\cS_{X_{AB}X_{BC}X_{CA}}\cS_{Y_{AB}Y_{BC}Y_{CA}}:=
\cS_{\max\{X_{AB},Y_{AB}\}
\max\{X_{BC},Y_{BC}\}\max\{X_{CA},Y_{CA}\}}$, where $\max\{X,Y\}$ is
the larger one between $X$ and $Y$ in the order $S \le P \le N \le D
\le M$. Therefore when $\ket{\Psi_1} \in {\cal
S}_{X_{AB}X_{BC}X_{CA}}$ and $\ket{\Psi_2} \in {\cal
S}_{Y_{AB}Y_{BC}Y_{CA}}$, the state $\ket{\Psi_1\cdot
\Psi_2}:=\ket{\Psi_1} \oplus \ket{\Psi_2} \in
\cS_{\max\{X_{AB},Y_{AB}\}
\max\{X_{BC},Y_{BC}\}\max\{X_{CA},Y_{CA}\}}$. This product is
commutative and in the direct sum, the subset $\cS_{SSS}$ is the
unit element but no inverse element exists. So the family of
non-empty sets $\cS_{X_{AB}X_{BC}X_{CA}}$ with the direct sum is an
abelian monoid, which is a commutative semigroup associated with the
unit.

The above analysis provides a systematic method to produce the
subsets in the monoid structure, except $\cS_{NMM}$ whose existence
is an open problem so far. Generally we have
$\cS_{SMM}=\cS_{SSM}\cS_{SMS}, \cS_{DDM} = \cS_{DDD}\cS_{SSM},
\cS_{DMM}=\cS_{DDD}\cS_{SMM}$, and $\cS_{MMM} = \cS_{PMM}\cS_{MSS}$.
So it is sufficient to check the non-emptyness of subsets
$\cS_{SSS}, \cS_{SSM}, \cS_{PMM}, \cS_{DDD}$. This fact has been
verified in Sec. III, and we can use the method to produce states in
$\cS_{DMM}$. The following is an example.


If we choose $\ket{\Psi_1} \in {\cal S}_{SMM}$ and $\ket{\Psi_2} \in
{\cal S}_{DDD}$ and both have $d_A=d_B=d_C $. Then the state
$\ket{\Psi_1\cdot \Psi_2}$ verifies the non-emptiness of the subset
${\cal S}_{DMM}$ with the condition  $r > d_C = d_A = d_B$. The
arguments have verified the existence of a boundary type mentioned
in the second paragraph below Lemma \ref{ha-l-2}.

\section{ generalization to multipartite system}


We begin by introducing the following results from
\cite{hst99,thapliyal99}. By fully separable states $\r$ of
$N$-partite systems, we mean $\r=\sum_i p_i \r_{1,i} \ox \cdots
\ox\r_{N,i}$.

 \bl
\label{le:PPTMxNrank<M,N} The $M\times N$ states of rank less than
$M$ or $N$ are distillable, and consequently they are NPT.
 \el

\begin{lemma}
 \label{le:Nfullyseparable=GHZ}
  The $N$-partite state $\ket{\ps}$ has $N$ fully separable $(N-1)$-partite reduced
  states if and only if $\ket{\ps}$ is a generalized GHZ state $\sum_i\sqrt{p_i}\ket{ii\cdots i}$ up to
  LU.
\end{lemma}

Next, we generalize Lemma \ref{le:SNSSSS} and
\ref{le:Nfullyseparable=GHZ}. For this purpose we define the
\emph{$n$-partite non-distillable state} $\r_{1\cdots n}$, in the
sense that one cannot distill any pure entangled state by collective
LOCC over any bipartition of parties $1, \cdots, m: (m+1),\cdots,n
$. By "collective LOCC" we regard parties $1, \cdots, m$ and
$(m+1),\cdots,n$ as two local parties, respectively. Similarly, we
define the \emph{$n$-partite PPT state} in the sense that any
bipartition of the state is PPT. Hence, such states contain a more
restrictive quantum correlation than the general multipartite state.
Evidently, the multipartite PPT state is a special multipartite
non-distillable state. The converse is unknown even for the
bipartite case \cite{dss00}. In what follows we will show the
equivalence of multipartite non-distillable, PPT and fully separable
states which are reduced states of a multipartite pure state.

For convenience we denote $\r_{\overline{i}}$ as a $(N-1)$-partite
state by tracing out the party $A_i$ from the $N$-partite state
$\ket{\ps}$, i.e., $\r_{\overline{i}} = \tr_i \proj{\ps}$. With
these definitions we have

\begin{theorem}
  \label{thm:ppt,n}
The following four statements are equivalent for a $N$-partite state
$\ket{\ps}$.
\\(1) $\ket{\ps}$ has $n$ non-distillable $(N-1)$-partite reduced
states $\r_{\overline{1}}, \cdots, \r_{\overline{n}}, ~N\ge n\ge2$;
\\(2) $\ket{\ps}$ has $n$ PPT $(N-1)$-partite reduced
states $\r_{\overline{1}}, \cdots, \r_{\overline{n}}, ~N\ge n\ge2$;
\\(3) $\ket{\ps}$ has $n$ fully separable $(N-1)$-partite reduced
states $\r_{\overline{1}}, \cdots, \r_{\overline{n}}, ~N\ge n\ge2$;
\\(4)  $\ket{\ps}=\sum_i \sqrt{p_i} \ket{i}^{\ox n} \ket{b_{i,n+1}} \ox \cdots \ox \ket{b_{i,N}} $ up to LU.
\end{theorem}
 \bpf
Suppose $\ket{\ps}$ is of dimensions $d_1\times \cdots \times d_N$.
The direction $(4)\ra(3)\ra(2)\ra(1)$ is evident. To show
$(1)\ra(4)$, suppose $\r_{\overline{1}}, \cdots, \r_{\overline{n}}$
are non-distillable. By using Lemma~\ref{le:PPTMxNrank<M,N} and the
definition of multipartite non-distillable states, we can obtain
$d:=d_1=\cdots=d_n \ge d_{n+1}, \cdots, d_N$. By combining the
parties $A_3, \cdots, A_N$ into one party, we obtain a tripartite
state satisfying Lemma \ref{le:SNSSSS}. Hence we have
$\ket{\ps}=\sum^d_{i=1} \sqrt{p_i} \ket{ii}\ket{\ph_i}$ where
$\ket{\ph_i} \in \bigox^N_{i=3} \cH_i$.

We show that $\ket{\ph_i}$ are fully product states. The proof is by
contradiction. Suppose there is, say $\ket{\ph_1}$ which is not
fully factorized. So we can write it as a bipartite entangled state
in the space $\cH_C \otimes \cH_D$. In other word we have the
4-partite state $\ket{\ps} \in \cH_1 \ox \cH_2 \ox \cH_C \ox \cH_D$
and the tripartite reduced state $\s_{2,C,D}=\sum^d_{i=1} \proj{i}_2
\ox \proj{\ph_i}_{CD}$. By performing the projector $\proj{1}$ on
space $\cH_2$, we can distill a pure entangled state $\ket{\ph_1}$
from $\s_{2,C,D}$. It contradicts with the assumption on
$\r_{\overline{1}}, \cdots,  \r_{\overline{n}}$. So every state
$\ket{\ph_i}$ has to be fully factorized and up to LU, we have
  \bea
  \ket{\ps}
  =
  \sum^d_{i=1} \sqrt{p_i} \ket{ii} \bigox^N_{j=3} \ket{a_{i,j}}.
  \eea

Next, we combine parties $A_1, A_4, \cdots, A_N$ together and make
$\ket{\ps}$ a new tripartite state. Because $\r_{\overline{3}}$ is
non-distillable and any entangled maximally correlated state is
distillable \cite{rains99}, the states $\ket{a_{i,3}}$ have to be
orthonormal. In a similar way we can prove the states
$\ket{a_{i,j}}, j=4, \cdots, n$ are orthonormal, respectively. So we
have justified the statement. This completes the proof of
$(1)\ra(4)$. So all four statements (1),(2),(3),(4) are equivalent.
 \epf

The following result is a stronger version of Lemma
\ref{le:Nfullyseparable=GHZ}.

\bcr
The $N$-partite state $\ket{\ps}$ has $N$ non-distillable
$(N-1)$-partite reduced states if and only if it is a generalized
GHZ state $\ket{\ps}=\sum_i \sqrt{p_i} \ket{i}^{\ox N} $ up to LU.
\ecr

\section{Conclusions}

We have proposed the converse monogamy of entanglement such that
when Alice and Bob are weakly entangled, then either of them is
generally strongly entangled with the third party. We believe that
the converse monogamy of entanglement is an essential quantum
mechanical feature and it promises a wide application in deciding
separability, entanglement distillation and quantum cryptography.
Our result presents two main open questions: First, can we propose a
concrete quantum scheme by the converse monogamy of entanglement?
Such a scheme will indicate a new essential difference between the
classical and quantum rules, just like that from quantum cloning
\cite{wz82} and the negative conditional entropy \cite{how05}.
Second, different from the monogamy of entanglement which relies on
the specific entanglement measures \cite{KW}, the converse monogamy
of entanglement only relies on the strength of entanglement. So can
we get a better understanding by adding other criteria on the
strength of entanglement such as the non-distillability ?

We also have shown tripartite pure states can be sorted into 21
subsets and they form an abelian monoid. It exhibits a more
canonical and clear algebraic structure of tripartite system
compared to the conventional SLOCC classification \cite{dvc00}. More
efforts from both physics and mathematics are required to understand
such structure.

\medskip

\acknowledgments

We thank Fernando Brandao for helpful discussion on the reduction
criterion and Andreas Winter for reading through the paper. The CQT
is funded by the Singapore MoE and the NRF as part of the Research
Centres of Excellence programme. MH is partially supported by a MEXT
Grant-in-Aid for Young Scientists (A) No. 20686026.

\section*{Appendix}

We prove Theorem \ref{thm:sixconditions} based on the following
preliminary lemma.
\begin{lemma}
\label{le:sep} Consider a tripartite state $\ket{\Psi_{ABC}}$ with a
separable reduced state $\rho_{BC}$. When $\rho_{AB}$ satisfies the
condition {\bf (6')}, it also satisfies the condition {\bf (1)}.
\end{lemma}
\begin{proof}
Due to separability, $\rho_{BC}$ can be written by
$\rho_{BC}=\sum_{i}p_i |\phi_i^B,\phi_i^C\rangle \langle
\phi_i^B,\phi_i^C|$. We introduce the new system ${\cal H}_D$ with
the orthogonal basis $e_i^D$ and the tripartite extension
$\rho_{BCD} := \sum_{i}p_i |\phi_i^B,\phi_i^C,e_i^D\rangle \langle
\phi_i^B,\phi_i^C,e_i^D|$. The monotonicity of the relative entropy
$D(\r||\s):=\tr (\r \log \r - \r \log\s)$ implies that
\begin{align*}
 0&=
 H(\rho_{C})-H(\rho_{BC})
 = D(\rho_{BC}\| I_B  \ox \rho_C )
 \nonumber\\
 &\le
 D(\rho_{BCD}\| I_B \ox \rho_{CD} )
 =\sum_{i}p_i (\log p_i -\log p_i)
=0,
\end{align*}
where the first equality is from Condition {\bf (6')}. So the
equality holds in the above inequality. According to Petz's
condition \cite{Petz}, the channel $\Lambda_C: \cH_C \mapsto \cH_C
\otimes \cH_D$ with the form $\Lambda_C(\sigma):=
\rho_{CD}^{1/2}((\rho_{C}^{-1/2}\sigma \rho_{C}^{-1/2}) \otimes
I_D)\rho_{CD}^{1/2} $ satisfies $ id_{B} \ox \Lambda_C
(\rho_{BC})=\rho_{BCD} $. We introduce the system $\cH_E$ as the
environment system of $\Lambda_C$ and the isometry
$U:\cH_C\mapsto\cH_C \otimes \cH_D \otimes \cH_E $ as the
Stinespring extension of $\Lambda_C$. So the state
$|\Phi_{ABCDE}\rangle:=I_{AB} \ox U\ket{\Psi_{ABC}}$ satisfies
$\rho_{BCD}=\Tr_{AE} \proj{\Phi_{ABCDE}}$. By using an orthogonal
basis $\{e^{AE}_i\}$ on $\cH_A \otimes \cH_E$ we can write up the
state $|\Phi_{ABCDE}\rangle = \sum_{i}\sqrt{p_i}
|\phi_i^B,\phi_i^C,e_i^D,e_i^{AE}\rangle $. Then, the state
$\r_{AB}=\tr_{CDE} \proj{\Phi_{ABCDE}} = \sum_{i}p_i
|\phi_i^B\rangle \langle \phi_i^B| \otimes \tr_{E}|e_i^{AE}\rangle
\langle e_i^{AE}|$ is separable. This completes the proof.
\end{proof}
Due to Lemma \ref{le:sep} and the equivalence of conditions {\bf
(5')} and {\bf (6')}, it suffices to show that when $\rho_{BC}$ is
non-distillable and $\rho_{AB}$ satisfies {\bf (5')}, $\rho_{BC}$ is
separable. From {\bf (5')} for $\rho_{AB}$, it holds that $\rank
\rho_{BC}=d_A=\rank \rho_{AB}=d_C$. It follows from \cite[Theorem
10]{cd11} that $\r_{BC}$ has to be PPT. So $\rho_{BC}$ is separable
by Lemma \ref{le:rankAB=rankA,B}, and we have Theorem
\ref{thm:sixconditions}.

\end{document}